\documentclass[a4paper,11pt]{article}
\pdfoutput=1 

\usepackage{jheppub} 

\usepackage{amsthm}
\newtheorem{proposition}{Proposition}[section]
\newtheorem{lemma}[proposition]{Lemma}

\usepackage[T1]{fontenc} 
\usepackage[polish]{babel}
\usepackage[utf8]{inputenc}

\title{\boldmath Pseudo-Riemannian Structures in Pati-Salam models}

\author{A. Bochniak,}
\author{T.E. Williams,}
\author{P. Zalecki}

\affiliation{Institute of Theoretical Physics, Jagiellonian University in Kraków,\\ prof. {\L}ojasiewicza 11, 30-348 Kraków, Poland}

\emailAdd{arkadiusz.bochniak@doctoral.uj.edu.pl}
\emailAdd{t.e.williams@doctoral.uj.edu.pl}
\emailAdd{pawel.zalecki@doctoral.uj.edu.pl}

\abstract{We discuss the role of the pseudo-Riemannian structure of the finite spectral triple for the family of Pati-Salam models. We argue that its existence is a very restrictive condition that separates leptons from quarks, and restricts the whole family of Pati-Salam models into the class of generalized Left-Right Symmetric Models.}

\keywords{Mathematical Methods of Physics, Symmetries}

\begin{document} 
\maketitle
\flushbottom
\section{Introduction}
The Standard Model of Particle Physics is the preeminent theory of fundamental constituents and their interactions. It has been experimentally tested and verified to be accurate with a very high degree of precision. Nevertheless, there are several open questions which, to date, the Standard Model has left unresolved.  Among them, several are related to the masses of neutrinos and the seesaw mechanism \cite{pr}, and others to baryon asymmetry \cite{sar}. Nowadays, there is also a tremendous set of cosmological data \cite{jarosik} which suggest the existence of dark matter. Several attempts to explain the aforementioned questions have already been proposed. Due to the proven success of the Standard Model, most of them are extensions thereof and are known as theories which go {\it Beyond the Standard Model}.  An interesting one, which has already been under consideration for several years and intensively studied by several physicists, is the model introduced by J.C. Pati and A. Salam \cite{ps}. 

The Pati-Salam model is a Yang-Mills-type model based on the 
$\mathrm{SU(2)}_R\times \mathrm{SU(2)}_L\times \mathrm{SU(4)}$
gauge group. It extends the usual Standard Model by e.g. introducing leptoquark symmetry and left-right symmetry. It was also considered to explain the origin of parity symmetry breaking \cite{mp}, \cite{mp1}.

Models used to describe theories of Particle Physics are traditionally constructed in the Lagrangian formalism, i.e. the form of the action is postulated based upon the desired symmetries of the resulting model. A possible geometrical explanation for the structure of such theories is provided by spectral geometry. The Standard Model has been studied for several years in the framework of almost-commutative geometry (see e.g. \cite{connes1},\cite{connes2} and \cite{connes3}) but some puzzles remain unsolved, not only for the product spectral triple, but also for its finite part. Recently, in \cite{bs}, it was proposed to consider the finite spectral triple for the Standard Model as the shadow of some pseudo-Riemannian triple in such a way that the pseudo-Riemannian structure leads to the existence of some nontrivial grading on the Riemannian triple.    

The Pati-Salam model was also considered as a noncommutative geometry by several authors, see e.g. \cite{amst},\cite{amst1},\cite{amst2},\cite{ccw1} and \cite{ccw}. We also highly recommend \cite{cw} as a survey of the historical development of these methods and their applications in Particle Physics. Reduced versions of the Pati-Salam models, i.e. the Left-Right Symmetric Models, were also considered in the framework of noncommutative geometry, firstly as potentially interesting examples for the Connes-Lott scheme of spectral geometry, but then also from the point of view of possible physical applications --- see e.g. \cite{gir},\cite{hj},\cite{is} and \cite{ok}. However, since some of the fundamental axioms of noncommutative geometry were not satisfied, such models were not satisfactory.  Later on, due to the trend of relaxing some of the axioms, e.g. the first order condition, and further development of the spectral theory in their absence (see e.g. \cite{ccw2}), the family of Pati-Salam models was analyzed.

In this paper we propose to take into account the pseudo-Riemannian structure of the finite triple for the Pati-Salam model, in the sense introduced in \cite{bs}. Herein, we analyze the finite spectral triples for such models, discuss possible pseudo-Riemannian structures, relate them to the gradation that distinguishes leptons from quarks, and argue that the existence of such a pseudo-Riemannian structure restricts the whole family of Pati-Salam models to models of the Left-Right Symmetric type.

\section{Finite spectral triples for Pati-Salam models}
In this section we consider the finite spectral triple for the family of Pati-Salam models. We discuss algebras and their commutants, different choices of chiral structures and possible Dirac operators.
\subsection{Spectral data}
The algebra for the Pati-Salam model is of the form 
\begin{equation}
\mathcal{A}=\mathbb{H}_R\oplus\mathbb{H}_L\oplus M_4(\mathbb{C}).
\end{equation}
Let $\mathcal{H}=F\oplus F^\ast$ be the Hilbert space for this model, where the elements $v$ of $F$ are presented in the following form
\begin{equation}
v=\begin{bmatrix}
\nu_R & u_R^1 & u_R^2 & u_R^3 \\
e_R & d_R^1 & d_R^2 & d_R^3 \\
\nu_L & u_L^1 & u_L^2 & u_L^3 \\
e_L & d_L^1 & d_L^2 & d_L^3
\end{bmatrix}. 
\end{equation}
Since
\begin{equation}
\mathrm{End}_{\mathbb{C}}(\mathcal{H})\cong M_4(\mathbb{C})\otimes M_2(\mathbb{C})\otimes M_4(\mathbb{C}),
\end{equation}
we have to represent all operators acting on $\mathcal{H}$ as elements of this tensor product space, and also find the form of the representation $\pi:\mathcal{A}\rightarrow \mathrm{End}_\mathbb{C}(\mathcal{H})$ in this language. Notice that $\mathrm{End}_{\mathbb{C}}(\mathcal{H})$ can be represented \cite{das} on $\mathcal{H}$ as
\begin{equation}
\widetilde{\pi}\left( \alpha \otimes 1_2\otimes\beta\right)\begin{bmatrix}
v \\w
\end{bmatrix} = \begin{bmatrix}
\alpha v\beta^t \\ 
\alpha w \beta^t
\end{bmatrix},
\end{equation}
and
\begin{equation}
\widetilde{\pi}\left(1_4\otimes\begin{bmatrix}
a & b\\
c& d
\end{bmatrix} \otimes 1_4 \right)\begin{bmatrix}
v \\w
\end{bmatrix}=\begin{bmatrix}
av+bw \\
cv+dw
\end{bmatrix},
\end{equation}
for all $\alpha,\beta,v,w\in M_4(\mathbb{C})$ and $a,b,c,d\in\mathbb{C}$.

Let $e_{ij}$ be the matrix that has $1$ in the entry $(i,j)$ and zero otherwise. Then the grading $\gamma$ has the following matrix representation
\begin{equation}
\label{gamma}
\gamma=\begin{bmatrix}
1_2 & \\
& -1_2
\end{bmatrix}\otimes e_{11}\otimes 1_4+1_4\otimes e_{22}\otimes \begin{bmatrix}
-1_2 & \\
& 1_2
\end{bmatrix}.
\end{equation}

There is also another possible choice of grading \cite{da}, for which left-handed leptons have the same parity as right-handed quarks, and vice versa for the opposite chirality:
\begin{equation}
\label{gamma_star}
\gamma_{\star}=\begin{bmatrix}
1_2 & \\
& -1_2
\end{bmatrix}\otimes e_{11} \otimes \begin{bmatrix}
1& \\ & -1_3 
\end{bmatrix} + \begin{bmatrix}
-1 & \\
 &1_3
\end{bmatrix}\otimes e_{22}\otimes \begin{bmatrix}
1_2 & \\
&-1_2
\end{bmatrix}.
\end{equation}

Let $J$ be the real structure, i.e. $J\begin{bmatrix}
v\\ w
\end{bmatrix}=\begin{bmatrix}
w^\ast \\ v^\ast
\end{bmatrix}$.
It is used to define the opposite representation \cite{das}. For $\xi=\widetilde{\pi}\left(\alpha\otimes\begin{bmatrix}
a & b\\
c& d
\end{bmatrix} \otimes \beta \right)$ we take $\xi^\circ=J\xi^\ast J^{-1}=\widetilde{\pi} \left(\beta^t\otimes\begin{bmatrix}
d & b\\
c& a
\end{bmatrix} \otimes \alpha^t \right)$. From now on we will omit the $\widetilde{\pi}$ symbol for a representation. 

The elements of the algebra $\mathcal{A}=\mathbb{H}_R\oplus\mathbb{H}_L\otimes M_4(\mathbb{C})$ are represented on $\mathcal{H}$ as
\begin{equation}
\pi(q_1,q_2,m)=\begin{bmatrix}
q_1 & \\
& q_2
\end{bmatrix}\otimes e_{11}\otimes 1_4 + m\otimes e_{22} \otimes 1_4,
\end{equation}
where $q_1\in\mathbb{H}_R$, $q_2\in\mathbb{H}_L$ and $m\in M_4(\mathbb{C})$.

Notice that $\gamma_{\star}$ does not commute with this representation of the algebra $\mathcal{A}$ unless the symmetry following from $ M_4(\mathbb{C})$ is broken into $\mathbb{C}\oplus M_3(\mathbb{C})$.

Therefore, here we are considering two algebras. The first one being $\mathbb{H}_R\oplus\mathbb{H}_L\oplus M_4(\mathbb{C})$ which we refer to as corresponding to an unreduced Pati-Salam model, and the second one $\mathbb{H}_R\oplus\mathbb{H}_L\oplus\mathbb{C}\oplus M_3(\mathbb{C})$, which we will call reduced.

Since the Dirac operator $D\in\mathrm{End}_\mathbb{C}(\mathcal{H})$, it is of the form
\begin{equation}
\label{dirac001}
D=\sum\limits_{ \substack{i,j=1,2  \\ 1\leq k\leq K} } D_{1ij}^k\otimes e_{ij}\otimes D_{2ij}^k,
\end{equation} 
with $D_{1ij}^k, D_{2ij}^k\in M_4(\mathbb{C})$, for some natural number $K$. From now on, we will assume that summations over omitted indices is over their entire range, if not explicitly stated otherwise. 
\subsection{Commutants}
We now consider the commutants of several algebras related to the unreduced and reduced Pati-Salam models. These results will be crucial to the discussion in section \ref{comms}.

Notice first, that for any (real or complex) matrix algebra $\mathcal{A}$ contained in $M_N(\mathbb{C})$ for some $N$, the commutant $\mathcal{A}'$ is the same as $\mathcal{A}_\mathbb{C}'$, where $\mathcal{A}_\mathbb{C}$ denotes the complexification of $\mathcal{A}$.

By a straightforward computation we can check that the commutant of the algebra of elements $\begin{bmatrix}
q_1 &\\
& q_2
\end{bmatrix}$, with $q_1,q_2\in \mathbb{H}$ is the algebra with elements $\begin{bmatrix}
\alpha 1_2&\\
& \beta 1_2
\end{bmatrix} $, where $\alpha,\beta\in\mathbb{C}$. We denote this algebra by $\mathcal{C}_1$.

In a similar manner, the commutant of the algebra of elements $\begin{bmatrix}
\lambda &\\
& n
\end{bmatrix}$, with $\lambda\in \mathbb{C}$, $n\in M_3(\mathbb{C})$ is the algebra with elements $\begin{bmatrix}
\alpha &\\
& \beta 1_3
\end{bmatrix} $, where $\alpha,\beta\in\mathbb{C}$. We denote this algebra by $\mathcal{C}_2$.

Furthermore, notice that $M_4(\mathbb{C})'\cong \mathbb{C}$. Therefore, we can describe the commutants of the Pati-Salam algebra for the unreduced (i.e. with $\mathcal{A}=\mathbb{H}_R\oplus\mathbb{H}_L\oplus M_4(\mathbb{C})$) case, and the reduced one (i.e. with $\mathcal{A}=\mathbb{H}_R\oplus\mathbb{H}_L\oplus \mathbb{C}\oplus M_3(\mathbb{C})$).

\begin{proposition}
\label{unreduced_comm}
The commutant $\left(\mathbb{H}_R\oplus\mathbb{H}_L\oplus M_4(\mathbb{C})\right)'$ in $\mathrm{End}_{\mathbb{C}}(\mathcal{H})$ is the algebra generated by elements of the form
\begin{equation}
A_1\otimes e_{11}\otimes E_1 +1_4\otimes e_{22}\otimes E_2,
\end{equation} 
where $A_1\in \mathcal{C}_1$ and $E_1,E_2\in M_4(\mathbb{C})$.
\end{proposition}
\begin{proof}
Any element of the considered algebra may be represented as
\begin{equation}
\pi(q_1,q_2,m)=\begin{bmatrix}
q_1 &\\
& q_2
\end{bmatrix}\otimes e_{11}\otimes 1_4 +m\otimes e_{22}\otimes 1_4,
\end{equation}
with $q_1,q_2\in\mathbb{H}$ and $m\in M_4(\mathbb{C})$. It is enough to find which elements of the form 
\begin{equation}
A_1\otimes e_{11}+A_2\otimes e_{12}+A_3\otimes e_{21}+A_4\otimes e_{22}\in M_4(\mathbb{C})\otimes M_2(\mathbb{C})
\end{equation}
commute with $\begin{bmatrix}
q_1 &\\
& q_2
\end{bmatrix}\otimes e_{11} +m\otimes e_{22}$ for all $q_1,q_2$ and $m$. 

The only possible solutions are with $A_1\in\mathcal{C}_1$, $A_4\sim 1_4$ and $A_2=A_3=0$.
\end{proof}

In a perfectly similar way, we get the following
\begin{proposition}
\label{reduced_comm}
The commutant $\left(\mathbb{H}_R\oplus\mathbb{H}_L\oplus \mathbb{C}\oplus M_3(\mathbb{C})\right)'$ in $\mathrm{End}_{\mathbb{C}}(\mathcal{H})$ is the algebra generated by elements of the form
\begin{equation}
A_1\otimes e_{11}\otimes E_1 +A_2\otimes e_{22}\otimes E_2,
\end{equation} 
where $A_1\in \mathcal{C}_1$, $A_2\in \mathcal{C}_2$ and $E_1,E_2\in M_4(\mathbb{C})$.
\end{proposition}

\subsection{Dirac operators and reality}
Let us  now consider self-adjoint Dirac operators that commute with the real structure $J$. Because $DJ=JD$, we have $D=JDJ^{-1}=\left(D^{\ast}\right)^\circ$, but since $D=D^\ast$, the necessary condition that has to be satisfied is $D = D^{\circ}$. Notice that this is a weaker condition than the original two together. Therefore, we will consider them separately. We begin with the following observation 
\begin{proposition}
\label{selfadj}
The Dirac operator $D=\sum D_{ijklrs}e_{kl}\otimes e_{ij}\otimes e_{rs}$ is self-adjoint if and only if $D_{ijklrs}=\bar{D}_{jilksr}$ for all indices.
\end{proposition}
\begin{proof}
Notice that
\begin{equation}D^\ast=\sum \bar{D}_{ijklrs} e_{lk}\otimes e_{ji}\otimes e_{sr}=\sum \bar{D}_{jilksr} e_{kl}\otimes e_{ij}\otimes e_{rs}.
\end{equation}
\end{proof}

Moreover, we have the following
\begin{proposition}
The Dirac operator $D=\sum D_{ijklrs}e_{kl}\otimes e_{ij}\otimes e_{rs}$ commutes with $J$ if and only if 
\begin{equation}
D_{11klrs}=\bar{D}_{22rskl},\qquad \mbox{for}\qquad D_{12klrs}=\bar{D}_{21rskl},
\end{equation}
for all $k,l,r,s$.
\end{proposition}
\begin{proof}
It follows from observing that for an operator $X$ of the form $A\otimes \begin{pmatrix}
a & b \\
c & d
\end{pmatrix}\otimes B $ we have 
\begin{equation}
JXJ^{-1}=\overline{B}\otimes \begin{pmatrix}
\bar{d} & \bar{c} \\
\bar{b} &\bar{a}
\end{pmatrix}\otimes \overline{A},
\end{equation}
and the fact that the requirement $DJ=JD$ is equivalent to $D=JDJ^{-1}$. 
\end{proof}

Let $DJ=JD$, i.e. $D=JDJ^{-1}$. We may write the Dirac operator as 
\begin{equation}
D=\underbrace{D_{11}+D_{12}}_{D_0}+\underbrace{D_{22}+D_{21}}_{D_1},\end{equation}
where 
\begin{equation}
D_{ij}=\sum\limits_k D_{1ij}^k\otimes e_{ij}\otimes D_{2ij}^k,
\end{equation}
as previously described in \eqref{dirac001}. It follows from the last proof that $D_1=JD_0J^{-1}$.

Let $A$ be an operator such that $AJ=\alpha JA$ for some $\alpha=\pm 1$. Moreover, suppose that $A=A_{11}+A_{22}$, where 
\begin{equation}
A_{ij}=\sum\limits_{k} A_{1ij}^k\otimes e_{ij}\otimes A_{2ij}^k,
\end{equation}
for some $A_{lij}^k\in M_4(\mathbb{C})$. Then we have the following
\begin{lemma}
$[A,D]=0$ if and only if $[A,D_0]=0$. Analogously for anticommutators.
\end{lemma}
\begin{proof}
Observe that $[A,D]=[A,D_0]+\alpha J[A,D_0]J$, since $J^2=\mathrm{id}$. Furthermore, $[A,D_0]$ only contains terms with $\cdots \otimes e_{11}\otimes \cdots$ and $\cdots \otimes e_{12} \otimes \cdots$, while $J[A,D_0]J$ only contains terms with $\cdots \otimes e_{21}\otimes \cdots$ and $\cdots \otimes e_{22}\otimes \cdots$.
\end{proof}
\subsection{Dirac operators for $\gamma$} Since we are interested in Dirac operators that commute with the real structure $J$ we can write $D=D_0+JD_0J^{-1}$. We first consider all Dirac operators $D$ that anticommute with $\gamma$, given by \eqref{gamma}, as a grading. It is enough to restrict ourselves to the $D_0$ part.

We have the following
\begin{proposition}
$D=D_0+JD_0J^{-1}$ anticommutes with $\gamma$ if and only if $D_0$ is of the form
\begin{equation}
\begin{split}
D_0&=\sum\limits_{k}\left\{\begin{bmatrix}
0_2& X_k \\
Y_k &0_2
\end{bmatrix}\otimes e_{11}\otimes A_k +  \begin{bmatrix}
P_k& Q_k \\
0_2 &0_2
\end{bmatrix}\otimes e_{12}\otimes \begin{bmatrix}
Z_k &0_2 \\
T_k& 0_2
\end{bmatrix} +\right. \\ &+ \left.  \begin{bmatrix}
0_2& 0_2\\
U_k& V_k 
\end{bmatrix}\otimes e_{12}\otimes \begin{bmatrix}
0_2& W_k \\
0_2& S_k
\end{bmatrix}  \right\},
\end{split}
\end{equation}
where $X_k, Y_k, Z_k, T_k, U_k, V_k, W_k, S_k,P_k,Q_k\in M_2(\mathbb{C})$ and  $A_k\in M_4(\mathbb{C})$.
\end{proposition} 
\begin{proof}
Notice that 
\begin{equation}
\gamma=\sum\limits_{n=1,2} (e_{nn}\otimes e_{11}\otimes e_{mm}-e_{mm}\otimes e_{22}\otimes e_{nn})+\sum\limits_{n=3,4}(e_{mm}\otimes e_{22}\otimes e_{nn}-e_{nn}\otimes e_{11}\otimes e_{mm}).
\end{equation}

Let us write $D$ as $D=\sum\hat{D}_{ijklrs},$ where $\hat{D}_{ijklrs}=D_{ijklrs}e_{kl}\otimes e_{ij}\otimes e_{rs}$. Simple computation shows that
\begin{equation}
D\gamma= \sum\limits_{n=1,2}\left(\hat{D}_{i1knrm}-\hat{D}_{i2kmrn}\right) + \sum\limits_{n=3,4}\left(\hat{D}_{i2kmrn}-\hat{D}_{i1knrm}\right),
\end{equation}
and
\begin{equation}
\gamma D=\sum\limits_{n=1,2}\left( \hat{D}_{1jnlms}-\hat{D}_{2jmlns}\right)+ \sum\limits_{n=3,4 }\left(\hat{D}_{2jmlns}-\hat{D}_{1jnlms}\right),
\end{equation}
hence by direct inspection we see that the Dirac operator has to be of the claimed form.
\end{proof}

\subsection{Dirac operators for $\gamma_\star$}
Again, we are considering all Dirac operators which commute with $J$, and therefore of the form $D=D_0+JD_0J^{-1}$, but now which anticommute with $\gamma_\star$, given by \eqref{gamma_star}, as a grading. This time we have,
\begin{proposition}
$D=D_0+JD_0J^{-1}$ anticommutes with $\gamma_\star$ if and only if $D_0$ is of the form
\begin{equation}
\begin{split}
D_0=& \sum_k\left\{\begin{bmatrix}
X_k &  \\
 &Y_k
\end{bmatrix}\otimes e_{11}\otimes \begin{bmatrix}
& \alpha_{1k} & \alpha_{2k} & \alpha_{3k} \\
\beta_{1k} &  &  &\\
\beta_{2k} &  &  &\\
\beta_{3k} &  &  &
\end{bmatrix}+\begin{bmatrix}
 & Z_k\\
T_k &
\end{bmatrix}\otimes e_{11}\otimes \begin{bmatrix}
\gamma_k & \\
& C_k
\end{bmatrix} +\right. \\  &+ \left.
\begin{bmatrix}
\delta_{1k} &\\
\delta_{2k} &\\
& P_k
\end{bmatrix}\otimes e_{12} \otimes \begin{bmatrix}
\sigma_{1k} & \sigma_{2k} &\\
& & E_k
\end{bmatrix}+\begin{bmatrix}
& F_k \\
\mu_{1k} &\\
\mu_{2k} &
\end{bmatrix}\otimes e_{12}\otimes \begin{bmatrix}
&\nu_{1k} &\nu_{2k}\\
G_k & &
\end{bmatrix} \right\},
\end{split}
\end{equation}
where $X_k,Y_k,Z_k,T_k\in M_2(\mathbb{C})$, $E_k,G_k\in M_{3\times 2}(\mathbb{C})$, $P_k,F_k\in M_{2\times 3}(\mathbb{C})$,  $C_k\in M_3(\mathbb{C})$, and $\alpha_{lk},\beta_{lk},\gamma_k, \delta_{lk},\sigma_{lk},\mu_{lk},\nu_{lk}\in \mathbb{C}$.
\end{proposition} 
\begin{proof}
As before, let us first write the grading in a more convenient form,
\begin{equation}
\begin{split}
\gamma_{\star}=& \sum\limits_{n=1,2}(e_{nn}\otimes e_{11}\otimes e_{11}-e_{11}\otimes e_{22}\otimes e_{nn})+\sum\limits_{\substack{n=3,4 \\ m=2,3,4}}(e_{nn}\otimes e_{11}\otimes e_{mm}-e_{mm}\otimes e_{22}\otimes e_{nn})+ \\
& +\sum\limits_{n=3,4}(e_{11}\otimes e_{22}\otimes e_{nn}-e_{nn}\otimes e_{11}\otimes e_{11})+\sum\limits_{\substack{m=1,2 \\ n=2,3,4}}(e_{nn}\otimes e_{22}\otimes e_{mm}-e_{mm}\otimes e_{11}\otimes e_{nn}).
\end{split}
\end{equation}
For $D=\sum\hat{D}_{ijklrs}$ notice that
\begin{equation}
\begin{split}
D\gamma_\star&= \sum\limits_{n=1,2 }\left(\hat{D}_{i1knr1}-\hat{D}_{i2k1rn}\right)+\sum\limits_{\substack{n=3,4 \\ m=2,3,4}}\left(\hat{D}_{i1knrm}-\hat{D}_{i2kmrn}\right)+\\
&+ \sum\limits_{n=3,4}\left(\hat{D}_{i2k1rn}-\hat{D}_{i1knr1}\right)+ \sum\limits_{\substack{m=1,2 \\ n=2,3,4 }}\left(\hat{D}_{i2knrm}-\hat{D}_{i1kmrn}\right),
\end{split}
\end{equation}
and likewise
\begin{equation}
\begin{split}
\gamma_\star D&= \sum\limits_{n=1,2}\left(\hat{D}_{1jnl1s}-\hat{D}_{2j1lns}\right) +\sum\limits_{\substack{n=3,4 \\ m=2,3,4 }}\left(\hat{D}_{1jnlms}-\hat{D}_{2jmlns}\right)+\\
&+ \sum\limits_{n=3,4}\left(\hat{D}_{2j1lns}-\hat{D}_{1jnl1s}\right)+\sum\limits_{\substack{m=1,2 \\ n=2,3,4 }}\left(\hat{D}_{2jnlms}-\hat{D}_{1jmlns}\right).
\end{split}
\end{equation}
Therefore, a straightforward comparison shows that $D$ anticommutes with $\gamma_\star$ if and only if it is of the claimed form.
\end{proof}

\section{Pseudo-Riemannian Structures}
Let us recall that a pseudo-Riemannian spectral triple $(\mathcal{A},\mathcal{H},\mathcal{D},\gamma, J,\beta)$, of signature $(p,q)$, is a system consisting of an algebra $\mathcal{A}$, Hilbert space $\mathcal{H}$, Dirac operator $\mathcal{D}$, $\mathbb{Z}/2\mathbb{Z}$-grading $\gamma$, real structure $J$ and an additional grading $\beta\in\mathrm{End}(\mathcal{H})$ such that $\beta^\ast=\beta$, $\beta^2=1$ and which commutes with the representation of $\mathcal{A}$ and defines a Krein structure on the Hilbert space. These objects are supposed to satisfy several conditions that are collected in \cite{bs}, section II.
For our purposes it is enough to recall that $\beta$ has to satisfy $\beta \gamma =(-1)^p\gamma \beta$ and $\beta J=(-1)^{\frac{p(p-1)}{2}}\epsilon^p J\beta$, where $\mathcal{D}J=\epsilon J\mathcal{D}$ and $\epsilon=\pm 1$ depending on the KO-dimension of the triple.

Furthermore, we assume that $\mathcal{D}$ is $\beta$-selfadjoint, i.e. $\mathcal{D}^\ast=(-1)^p\beta \mathcal{D}\beta$. We say that the triple is time-oriented if $\beta$ can be presented as the image of a certain Hochschild $p$-cycle. 

Out of the pseudo-Riemannian spectral triple $(\mathcal{A},\mathcal{H},\mathcal{D},\gamma, J,\beta)$ one can construct its Riemannian restriction, i.e. a triple $(\mathcal{A},\mathcal{H},D_+,\gamma, J,\beta)$ with $\mathcal{D}_+=\frac{1}{2}(\mathcal{D}+\mathcal{D}^\ast)$ which is a \mbox{self-adjoint} operator and $\beta \mathcal{D}_+=(-1)^p\mathcal{D}_+\beta$. This spectral triple is of the same KO-dimension as the one we started with.

As was noticed in \cite{bs} the spectral triple for the Standard Model can be treated as a Riemannian restriction of some pseudo-Riemannian triple, with an additional grading originating from the time-orientation on the pseudo-Riemannian level. The existence of such a gradation results in a restriction on the number of possible Dirac operators, compatible with the other elements of the triple. 

Here we are looking for similar effects in the case of Pati-Salam models. From now on we will denote $\mathcal{D}_+$ by $D$, and since the spectral triple for the Pati-Salam models has to be of KO-dimension $6$, we take the signature to be $(0,2)$. We are looking for all possible $\beta$s and (self-adjoint) Dirac operators $D$ such that $\beta D=D\beta$.

Therefore we are looking for $\beta$ of the form 
\begin{equation}
\beta=\pi(q_1,q_2,m)J\pi(q_1',q_2',m')^\ast J^{-1},
\end{equation}
with $q_1,q_2,q_1',q_2'\in\mathbb{H}$ and $m,m'\in M_4(\mathbb{C})$ for the unreduced case, and 
\begin{equation}
\beta=\pi(q_1,q_2,\lambda,n)J\pi(q_1',q_2',\lambda',n')^\ast J^{-1},
\end{equation}
with $q_1,q_2,q_1',q_2'\in\mathbb{H}$, $\lambda,\lambda'\in \mathbb{C}$ and $n,n'\in M_3(\mathbb{C})$ in the reduced case.

For simplicity we have assumed here that $\beta$ is $0$-cycle containing only one summand. We postpone the discussion about the more general situation until section \ref{comms}.

Moreover, we require that 
\begin{equation}
\beta\gamma=\gamma\beta, \qquad \beta J=J\beta.
\end{equation}
\subsection{The unreduced Pati-Salam model \label{unreducedPS}}
Let us start with this case first. Then $\beta$ can be represented as 
\begin{equation}
\beta=\begin{bmatrix}
q_1 & \\
& q_2
\end{bmatrix}\otimes e_{11}\otimes m'^t +m\otimes e_{22}\otimes\begin{bmatrix}
q_1'^t & \\
& q_2'^t
\end{bmatrix}.
\end{equation}
Since $\beta$ is a $0$-cycle, it commutes with the grading by construction. Notice the fact that $\beta$ commutes both with the algebra and the opposite algebra (since $\beta J=J\beta$ and the order zero condition holds) fixes all matrices $q_1,q_2,q_1',q_2'$ and $m,m'$ to be proportional to the identity, and moreover, it enforces $m=m'^t$ and $q_i=q_i'^t$ for $i=1,2$. The condition $\beta^2=1$ fixes all these proportionality factors to be a sign. 

Therefore, the only possible pseudo-Riemannian structures are
\begin{equation}
\beta=\pi(\eta_1 1_2,\eta_2 1_2,\eta_3 1_4)J\pi(\eta_1 1_2,\eta_2 1_2,\eta_3 1_4)J^{-1}
\end{equation}
with $\eta_1,\eta_2,\eta_3=\pm1$. 

So, up to the trivial rescaling by a factor of $-1$ there are only two such possible operators:
\begin{equation}
\pi(1_2,1_2,1_4)J\pi(1_2,1_2,1_4)J^{-1}, \quad \mbox{and} \quad  \pi(1_2,-1_2,1_4)J\pi(1_2,-1_2,1_4)J^{-1}.
\end{equation}

\subsubsection{Compatible Dirac operators for the unreduced Pati-Salam model}

We are looking for all possible generic Dirac operators $D$ (not necessarily anticommuting with a grading) such that  $D\beta=\beta D$. Moreover, we already assume that $D$ commutes with $J$, so it is of the form $D_0+JD_0J^{-1}$. For such a $D$, in each of the cases for $\beta$, we get the following restrictions.

First, notice that $\beta=\pi(1_2,1_2,1_4)J\pi(1_2,1_2,1_4)J^{-1}=1_{32}$ is the identity operator, so it commutes with everything. This is the trivial case in which we are not interested.

For the second choice we get the following
\begin{proposition}
\label{nontrivial}The Dirac operator $D=D_0+JD_0J^{-1}$ commutes with 
\begin{equation}
\beta=\pi(1_2,-1_2,1_4)J\pi(1_2,-1_2,1_4)J^{-1}
\end{equation} 
if and only if $D_0$ is of the following form
\begin{equation}
\begin{split}
D_0&=\sum\limits_{k}\left\{\begin{bmatrix}
 \widetilde{X}_k &0_2\\
0_2&\widetilde{Y}_k 
\end{bmatrix}\otimes e_{11}\otimes \widetilde{A}_k + \begin{bmatrix}
\widetilde{P}_k& \widetilde{Q}_k \\
 0_2&0_2
\end{bmatrix}\otimes e_{12}\otimes \begin{bmatrix}
\widetilde{Z}_k &0_2 \\
\widetilde{T}_k&0_2
\end{bmatrix}  + \right. \\ &+ \left.   \begin{bmatrix}
0_2& 0_2\\
\widetilde{U}_k& \widetilde{V}_k 
\end{bmatrix}\otimes e_{12}\otimes \begin{bmatrix}
0_2& \widetilde{W}_k \\
0_2& \widetilde{S}_k
\end{bmatrix} \right\},
\end{split}
\end{equation}
where $\widetilde{X}_k, \widetilde{Y}_k, \widetilde{P}_k,\widetilde{Q}_k, \widetilde{Z}_k, \widetilde{T}_k, \widetilde{U}_k, \widetilde{V}_k, \widetilde{W}_k, \widetilde{S}_k\in M_2(\mathbb{C})$ and  $\widetilde{A}_k\in M_4(\mathbb{C})$.
\end{proposition}
\begin{proof}
We first write
\begin{equation}
\begin{split}
\beta=& \sum\limits_{n=1,2}e_{nn}\otimes e_{11}\otimes e_{mm}-\sum\limits_{n=3,4} e_{nn}\otimes e_{11}\otimes e_{mm}+\\
+&\sum\limits_{m=1,2} e_{nn}\otimes e_{22}\otimes e_{mm}-\sum\limits_{m=3,4} e_{nn}\otimes e_{22}\otimes e_{mm}.
\end{split}
\end{equation}
Noticing that for $D=\sum \hat{D}_{ijklrs}$ we have
\begin{equation}
D\beta=\sum\limits_{n=1,2}\hat{D}_{i1knrm}-\sum\limits_{n=3,4}\hat{D}_{i1knrm}+ \sum\limits_{m=1,2}\hat{D}_{i2knrm} -\sum\limits_{m=3,4}\hat{D}_{i2knrm},
\end{equation}
and similarly
\begin{equation}
\beta D=\sum\limits_{n=1,2}\hat{D}_{1jnlms}-\sum\limits_{n=3,4}\hat{D}_{1jnlms}+\sum\limits_{m=1,2}\hat{D}_{2jnlms} - \sum\limits_{m=3,4}\hat{D}_{2jnlms}.
\end{equation}
The result follows from a straightforward comparison of these expressions. 
\end{proof}
\subsubsection{Physical consequences of the unreduced Pati-Salam model with $\gamma$}

Since $D$ anticommutes with $\gamma$, we see that the only possibility for such a $D$ to commute with the nontrivial $\beta$ discussed above is that terms of the form $\cdots \otimes e_{11}\otimes \cdots$ and $\cdots \otimes e_{22} \otimes \cdots$ must vanish. Notice that only models which contain these terms are physically acceptable as extensions of the Standard Model, for which case the Dirac operator has to contain terms $\begin{bmatrix}
&M_l \\
M_l^\dagger 
\end{bmatrix}\otimes e_{11}\otimes e_{11}$ and $\begin{bmatrix}
&M_q \\
M_q^\dagger 
\end{bmatrix}\otimes e_{11}\otimes (1- e_{11})$, which encode the Yukawa parameters for leptons and quarks. Therefore, no Pati-Salam model with the algebra $\mathbb{H}_R\oplus\mathbb{H}_L\oplus M_4(\mathbb{C})$, grading $\gamma$, and with the pseudo-Riemannian structure $\beta$ is physically acceptable.

\subsection{The reduced Pati-Salam model  \label{reducedPS}}
In this case
\begin{equation}
\beta= \begin{bmatrix}
q_1 & \\
& q_2
\end{bmatrix}\otimes e_{11}\otimes\begin{bmatrix}
\lambda'& \\
& n'^t
\end{bmatrix}  + \begin{bmatrix}
\lambda& \\
& n
\end{bmatrix}\otimes e_{22}\otimes \begin{bmatrix}
q_1'^t & \\
& q_2'^t
\end{bmatrix}.
\end{equation}
As before, since $\beta$ is a $0$-cycle, it commutes with the grading, and 
commutation with the algebra, the fact that $\beta^2=1$ and that it commutes with $J$ fixes all matrices $q_i=q_i'=\pm 1_2$, for $i=1,2$, $\lambda=\lambda'=\pm 1$ and $n=n'=\pm 1_3$. Therefore, up to trivial rescaling we have the following four cases
\begin{equation}
\begin{split}
\pi(1_2,1_2,1,1_3)J\pi(1_2,1_2,1,1_3)J^{-1}, & \quad  \pi(1_2,-1_2,1,1_3)J\pi(1_2,-1_2,1,1_3)J^{-1},\\
\pi(1_2,1_2,1,-1_3)J\pi(1_2,1_2,1,-1_3)J^{-1}, &\quad  \pi(-1_2,1_2,1,-1_3)J\pi(-1_2,1_2,1,-1_3)J^{-1}.
\end{split}
\end{equation}
Only three of them are nontrivial.
\subsubsection{Compatible Dirac operators for the reduced Pati-Salam model}

Now, we will discuss restrictions on a generic Dirac operator (not necessary anticommuting with a grading) which follow from commutation with the nontrivial $\beta$s allowed in the case of the reduced Pati-Salam model. As before, we assume that $D$ commutes with the real structure so that it is of the form $D_0+JD_0J^{-1}$. The case with $\beta=\pi(1_2,-1_2,1,1_3)J\pi(1_2,-1_2,1,1_3)J^{-1}$ is exactly the same as the one discussed in subsection \ref{nontrivial}. For the two remaining cases we get the following results. Firstly, we have 
\begin{proposition}
\label{case2}
The Dirac operator $D=D_0+JD_0J^{-1}$ commutes with 
\begin{equation}
\beta=\pi(-1_2,1_2,1,-1_3)J\pi(-1_2,1_2,1,-1_3)J^{-1}
\end{equation} if and only if $D_0$ is of the form 
\begin{equation}
\begin{split}
D_0=& \sum_k\left\{\begin{bmatrix}
\widetilde{X}_k &  \\
 &\widetilde{Y}_k
\end{bmatrix}\otimes e_{11}\otimes  \begin{bmatrix}
\widetilde{\gamma}_k & \\
& \widetilde{C}_k
\end{bmatrix}+\begin{bmatrix}
 & \widetilde{Z}_k\\
\widetilde{T}_k &
\end{bmatrix}\otimes e_{11}\otimes \begin{bmatrix}
& \widetilde{\alpha}_{1k} & \widetilde{\alpha}_{2k} & \widetilde{\alpha}_{3k} \\
\widetilde{\beta}_{1k} &  &  &\\
\widetilde{\beta}_{2k} &  &  &\\
\widetilde{\beta}_{3k} &  &  &
\end{bmatrix} +\right. \\  &+ \left.
\begin{bmatrix}
\widetilde{\delta}_{1k} &\\
\widetilde{\delta}_{2k} &\\
& \widetilde{P}_k
\end{bmatrix}\otimes e_{12} \otimes \begin{bmatrix}
\widetilde{\sigma}_{1k} & \widetilde{\sigma}_{2k} &\\
& & \widetilde{E}_k
\end{bmatrix}+\begin{bmatrix}
& \widetilde{F}_k \\
\widetilde{\mu}_{1k} &\\
\widetilde{\mu}_{2k} &
\end{bmatrix}\otimes e_{12}\otimes \begin{bmatrix}
&\widetilde{\nu}_{1k} &\widetilde{\nu}_{2k}\\
\widetilde{G}_k & &
\end{bmatrix}  \right\},
\end{split}
\end{equation}
where $\widetilde{X}_k,\widetilde{Y}_k,\widetilde{Z}_k,\widetilde{T}_k\in M_2(\mathbb{C})$, $\widetilde{E}_k,\widetilde{G}_k\in M_{3\times 2}(\mathbb{C})$, $\widetilde{P}_k,\widetilde{F}_k\in M_{2\times 3}(\mathbb{C})$,  $\widetilde{C}_k\in M_3(\mathbb{C})$, and $\widetilde{\alpha}_{lk},\widetilde{\beta}_{lk},\widetilde{\gamma}_k, \widetilde{\delta}_{lk},\widetilde{\sigma}_{lk},\widetilde{\mu}_{lk},\widetilde{\nu}_{lk}\in \mathbb{C}$.
\end{proposition}
\begin{proof}
It follows from a straightforward computation that, for $D=\sum \hat{D}_{ijklrs}$, 
\begin{equation}
\begin{split}
D\beta=&-\sum\limits_{n=1,2}\hat{D}_{i1knr1}+\sum\limits_{\substack{n=1,2 \\ m=2,3,4}}\hat{D}_{i1knrm}+\sum\limits_{n=3,4}\hat{D}_{i1knr1}-\sum\limits_{\substack{n=3,4 \\ m=2,3,4}}\hat{D}_{i1knrm}-\\
&-\sum\limits_{m=1,2}\hat{D}_{i2k1rm}+\sum\limits_{m=3,4}\hat{D}_{i2k1rm}+\sum\limits_{\substack{n=2,3,4 \\ m=1,2}}\hat{D}_{i2knrm}-\sum\limits_{\substack{n=2,3,4\\ m=3,4}}\hat{D}_{i2knrm},
\end{split}
\end{equation}
and
\begin{equation}
\begin{split}
\beta D=&-\sum\limits_{n=1,2}\hat{D}_{1jnl1s}+\sum\limits_{\substack{n=1,2 \\ m=2,3,4}}\hat{D}_{1jnlms}+\sum\limits_{n=3,4}\hat{D}_{1jnl1s}-\sum\limits_{\substack{n=3,4 \\ m=2,3,4}}\hat{D}_{1jnlms} -\\
&-\sum\limits_{m=1,2}\hat{D}_{2j1lms}+\sum\limits_{m=3,4}\hat{D}_{2j1lms}+\sum\limits_{\substack{n=2,3,4 \\ m=1,2}}\hat{D}_{2jnlms}-\sum\limits_{\substack{n=2,3,4\\ m=3,4}}\hat{D}_{2jnlms}.
\end{split}
\end{equation}
Comparing these two expressions we get the claimed result.
\end{proof}

Similarly,
\begin{proposition}
\label{final}
The Dirac operator $D=D_0+JD_0J^{-1}$ commutes with \begin{equation}
\beta=\pi(1_2,1_2,1,-1_3)J\pi(1_2,1_2,1,-1_3)J^{-1}
\end{equation} if and only if $D_0$ is of the form
\begin{equation}
\begin{split}
D_0=& \sum_k\left\{\widetilde{A}^{(11)}_k\otimes e_{11}\otimes \begin{bmatrix}
\widetilde{\gamma}_k & \\
& \widetilde{C}_k
\end{bmatrix}  +
\begin{bmatrix}
\widetilde{\delta}_{1k} &&&\\
\widetilde{\delta}_{2k} &&0_{4\times3}&\\
\widetilde{\delta}_{3k} &&&\\
\widetilde{\delta}_{4k} &&&
\end{bmatrix}\otimes e_{12} \otimes \begin{bmatrix}
\widetilde{\sigma}_{1k} & \widetilde{\sigma}_{2k} &\widetilde{\sigma}_{3k} & \widetilde{\sigma}_{4k} \\
&&&\\
&&0_{3\times 4}&\\
&&&
\end{bmatrix}+\right.\\ &+\left.\begin{bmatrix}
0_{4\times 1}& \widetilde{F}_k 
\end{bmatrix}\otimes e_{12}\otimes \begin{bmatrix}
0_{1\times 4}\\
\widetilde{G}_k 
\end{bmatrix} \right\},
\end{split}
\end{equation}
where $\widetilde{A}_k\in M_4(\mathbb{C})$, $\widetilde{G}_k\in M_{3\times 4}(\mathbb{C})$, $\widetilde{F}_k\in M_{4\times 3}(\mathbb{C})$,  $\widetilde{C}_k\in M_3(\mathbb{C})$, and $\widetilde{\gamma}_k, \widetilde{\delta}_{lk},\widetilde{\sigma}_{lk}\in \mathbb{C}$.
\end{proposition}
\begin{proof}
In a similar manner to the previous cases we compute
\begin{equation}
D\beta=\sum \left(\hat{D}_{i1klr1}+\hat{D}_{i2k1rs}\right) -\sum\limits_{m=2,3,4}\left(\hat{D}_{i1klrm}+\hat{D}_{i2knrs}\right),
\end{equation}
and
\begin{equation}
\beta D= \sum \left(\hat{D}_{1jkl1s}+\hat{D}_{2j1lrs}\right)-\sum\limits_{m=2,3,4}\left(\hat{D}_{1jklms}+\hat{D}_{2jnlrs}\right).
\end{equation}
The result follows from a straightforward comparison of these terms.
\end{proof}

\subsubsection{Physical consequences of the reduced Pati-Salam model with $\gamma_\ast$}

Notice that since we require the Dirac operator to anticommute with the grading $\gamma_{\ast}$, and moreover any physically interesting model should be an extension of the Standard Model, we conclude that the only possibility is therefore the $\beta$ from Proposition \ref{final}, in which case anticommutation with $\gamma_{\ast}$ reduces the freedom of possible Dirac operators $D=D_0+JD_0J^{-1}$ to those with $D_0$ of the form
\begin{equation}
\begin{split}
D_0&=\sum\limits_{k}\left\{ \begin{bmatrix}
& Z_k\\
T_k &
\end{bmatrix}\otimes e_{11}\otimes \begin{bmatrix}
\gamma_k & \\
& C_k
\end{bmatrix}  
\right. +\\
&+\left. \begin{bmatrix}
\delta_{1k} & & \\
\delta_{2k} & 0_{4\times 3} \\
0&&\\
0&&
\end{bmatrix}\otimes e_{12}\otimes \begin{bmatrix}
\sigma_{1k}& \sigma_{2k} &0 &0 \\
&&&\\
&0_{3\times 4}&&
\end{bmatrix} +
\begin{bmatrix}
0_{2\times 1} & 0_{2\times 3} \\
0_{2\times 1}& E_{1k}
\end{bmatrix}\otimes e_{12}\otimes\begin{bmatrix}
0_{1\times 2} & 0_{1\times 2} \\
0_{3\times 2}& F_{1k}
\end{bmatrix}\right.\\
&+\left.
\begin{bmatrix}
0 & & \\
0 & 0_{4\times 3} \\
\delta_{3k}&&\\
\delta_{4k}&&
\end{bmatrix}\otimes e_{12}\otimes \begin{bmatrix}
0&0&\sigma_{3k} & \sigma_{4k}\\
&&&\\
&&0_{3\times 4}&
\end{bmatrix}+\begin{bmatrix}
0_{2\times 1} & E_{2k}\\
0_{2\times 1}& 0_{2\times 3}
\end{bmatrix}\otimes e_{12}\otimes\begin{bmatrix}
0_{1\times 2} & 0_{1\times 2} \\
F_{2k}&0_{3\times 2}
\end{bmatrix}\right\},
\end{split}
\label{Dirac_gamma}
\end{equation} 
where $T_k,Z_k\in M_2(\mathbb{C})$, $C_{k}\in M_3(\mathbb{C})$, $E_{lk}\in M_{2\times 3}(\mathbb{C})$, $F_{lk}\in M_{3\times 2}(\mathbb{C})$ and $\gamma_k,\sigma
_{lk},\delta_{lk}\in\mathbb{C}$.

We can treat this model as an extension of the Standard Model with modified chiralities, i.e. in which left-handed (resp. right-handed) leptons have the same parity as right-handed (resp. left-handed) quarks. Therefore, the only compatible extension beyond the Standard Model and contained within the family of Pati-Salam models which have $\gamma_\ast$ as a grading, and possesses a pseudo-Riemannian structure in the sense of the existence of a one-term $0$-cycle $\beta$, is precisely the reduced Pati-Salam model with exactly the same pseudo-Riemannian structure which was uniquely possible in the case of the Standard Model \cite{bs}. Since $\gamma_\ast$ explicitly breaks the $\mathrm{SU(4)}$-symmetry into $\mathrm{U(1)}\times \mathrm{U(3)}$, it is not surprising that the resulting class of models also has this property. Nevertheless, it is worth noting that there is exactly one (up to an irrelevant global sign in $\beta$) such possibility, and moreover that it is consistent with the one that is known to be the only possibility in the case of the Standard Model. Furthermore, we observe that the pseudo-Riemannian structure still allows for the existence of $\mathrm{SU(2)}$-doublets of right-handed particles. This is an interesting feature. Notice also, that there are further restrictions on the entries of the compatible Dirac operators which follow from the assumption of self-adjointness. These restrictions are summarized in Proposition \ref{selfadj}.
\subsubsection{Physical consequences of the reduced Pati-Salam model with $\gamma$}

Since the grading $\gamma$ is compatible with the unreduced Pati-Salam model, it is also compatible with the reduced model. Therefore we can consider Dirac operators that anticommute with $\gamma$ and commute with $\beta$ for this case. For the nontrivial $\beta$s we see that only one of them, i.e. $\beta$ from Proposition \ref{final}, is compatible with the requirement of being an extension of the Standard Model. For it we have the following
 \begin{proposition}
 The Dirac operator $D=D_0+JD_0J^{-1}$ that commutes with $\beta$ and anticommutes with $\gamma$ has to be of the same form as in \eqref{Dirac_gamma}, i.e. exactly the same form as in the case with $\gamma_\ast$.
\end{proposition} 
 Moreover, notice that the only $\beta$ which is admissible in this case is exactly the same one that prevented the existence of leptoquarks in the Standard Model \cite{bs}. Therefore, the only possible extension of the Standard Model contained within the family of Pati-Salam models, which takes into account the pseudo-Riemannian structure for finite triples in the sense defined in \cite{bs}, has to be of the reduced form. That is, the $\mathrm{SU(4)}$-symmetry is broken into a $\mathrm{U(1)}\times \mathrm{U(3)}$-symmetry. Therefore, instead of the full $\mathrm{SU(2)}_R\times \mathrm{SU(2)}_L\times \mathrm{SU(4)}$ Pati-Salam gauge group we must reduce to the case with $\mathrm{SU(2)}_R\times \mathrm{SU(2)}_L\times \mathrm{U(1)}\times \mathrm{U(3)}$. In this extension the right particles are doublets under the $\mathrm{SU(2)}$-symmetry, and leptons are separated from quarks, i.e. they are not the fourth color, so there are no leptoquarks in the sense of degrees of freedom. In a similar manner to the previous case, there are further restrictions on the entries of the compatible Dirac operators which follow from the assumption of self-adjointness --- see Proposition \ref{selfadj}.
\subsection{Generic $\beta$-structures}
\label{comms}
Here we are looking for all possible $\beta$s that are $0$-cycles and which satisfy all required conditions but we do not assume these operators to consist of only one term.
\subsubsection{The unreduced Pati-Salam model}
Since we require that $\beta$ commutes with the representation of the unreduced Pati-Salam algebra it follows from Proposition \ref{unreduced_comm} that 
\begin{equation}
\beta=\begin{bmatrix}
1_2 & \\
& 0_2
\end{bmatrix}\otimes e_{11}\otimes E_1 + \begin{bmatrix}
0_2 & \\
& 1_2
\end{bmatrix}\otimes e_{11} \otimes E_2 +1_4\otimes e_{22}\otimes F,
\end{equation} 
where $E_1,E_2,F\in M_4(\mathbb{C})$. 

Since $\beta^2=1$, we get
\begin{equation}
F^2=1_4, \qquad\begin{bmatrix}
1_2&\\ 
& 0_2
\end{bmatrix}\otimes E_1^2+\begin{bmatrix}
0_2&\\ 
& 1_2
\end{bmatrix}\otimes E_2^2=1_{16},
\end{equation}
hence $E_1^2=E_2^2=1_4$.
Since $\beta$ commutes with $J$, this implies that 
\begin{equation}
\label{cond1}
\bar{F}\otimes 1_4= \begin{bmatrix}
1_2&\\ 
& 0_2
\end{bmatrix}\otimes E_1+ \begin{bmatrix}
0_2&\\ 
& 1_2
\end{bmatrix}\otimes E_2. 
\end{equation}
Let us write $F=\begin{bmatrix}
F^{11} & F^{12}\\
F^{21} & F^{22}
\end{bmatrix} $, then \eqref{cond1} is equivalent to the following set of conditions
\begin{equation}
\begin{bmatrix}
\overline{F_{11}} & \overline{F_{12}}\\
\overline{F_{21}} & \overline{F_{22}}
\end{bmatrix}\otimes 1_4=\begin{bmatrix}
1_2\otimes E_1 & \\
& 1_2\otimes E_2
\end{bmatrix}
\end{equation}
Therefore $\overline{F_{12}}=\overline{F_{21}}=0_2$, $\overline{F_{11}}\otimes 1_4=1_2\otimes E_1$ and $\overline{F_{22}}\otimes 1_4=1_2\otimes E_2$, so 
\begin{equation}
E_1=\eta_11_4, \qquad \overline{F_{11}}=\eta_1 1_2, \qquad E_2=\eta_2 1_4, \qquad \overline{F_{22}}=\eta_2 1_2,
\end{equation}
for some nonzero complex numbers $\eta_{1}$ and $\eta_{2}$. Notice that since $\beta^\ast =\beta$ and $\beta^2=1$ the zero solutions are not allowed, and moreover we deduce that both $\eta_{1}=\pm 1$ and $\eta_{2}=\pm 1$. Therefore,
\begin{equation}
\beta=\begin{bmatrix}
\eta_1 1_2 & \\
&\eta_2 1_2
\end{bmatrix}\otimes e_{11}  \otimes 1_4 +1_4\otimes e_{22}\otimes \begin{bmatrix}
\eta_11_2 & \\
& \eta_21_2
\end{bmatrix},
\end{equation}
for some $\eta_1,\eta_2$ being $\pm 1$. 

Notice that all such $\beta$s are $0$-cycles, and there are (up to a trivial rescaling) only two possibilities:
\begin{equation}
\begin{split}
\pi(1_2,1_2,1_4)J\pi(1_2,1_2,1_4)J^{-1}, &\quad  \pi(1_2,-1_2,1_4)J\pi(1_2,-1_2,1_4)J^{-1}.
\end{split}
\end{equation}
These are exactly the same as under the assumption of $\beta$ being only a one-term $0$-cycle, i.e. the conclusion of section \ref{unreducedPS} remains valid for more general $\beta$s.
\subsubsection{The reduced Pati-Salam model}
In a similar manner, since we require that $\beta$ commutes with the representation of the reduced Pati-Salam algebra, it follows from Proposition \ref{reduced_comm} that 
\begin{equation}
\beta=\begin{bmatrix}
1_2& \\
& 0
\end{bmatrix}\otimes e_{11}\otimes E_1 +\begin{bmatrix}
0& \\
&1_2
\end{bmatrix}\otimes e_{11}\otimes E_2+\begin{bmatrix}
1&\\
& 0_3
\end{bmatrix}\otimes e_{22}\otimes F_1 +\begin{bmatrix}
0 & \\
& 1_3
\end{bmatrix}\otimes e_{22}\otimes F_2,
\end{equation}
where $E_1,E_2,F_1,F_2\in M_4(\mathbb{C})$. Since $\beta^2=1$, we infer that
\begin{equation}
E_1^2=E_2^2=1_4=F_1^2=F_2^2.
\end{equation}
Now, from the condition $\beta J=J \beta$, repeating the previously used argument, we end up with the following form of $\beta$:
\begin{equation}
\begin{split}
\beta &=\begin{bmatrix}
1_2 & \\
& 0_2
\end{bmatrix} \otimes e_{11} \otimes \begin{bmatrix}
\eta_1 & \\
& \eta_2 1_3
\end{bmatrix} +\begin{bmatrix}
0_2 &\\
& 1_2
\end{bmatrix}\otimes e_{11}\otimes \begin{bmatrix}
\eta_3 & \\
& \eta_4 1_3
\end{bmatrix}+ \\
&+\begin{bmatrix}
1& \\
& 0_3
\end{bmatrix}\otimes e_{22}\otimes \begin{bmatrix}
\eta_1 1_2 & \\
& \eta_3 1_2
\end{bmatrix} +\begin{bmatrix}
0 &\\
& 1_3
\end{bmatrix}\otimes e_{22}\otimes \begin{bmatrix}
\eta_2 1_2 & \\
& \eta_4 1_2
\end{bmatrix},
\end{split}
\end{equation} 
where $\eta_i=\pm 1$ for $i=1,...,4$. There are only eight independent (i.e. up to a global sign) possibilities. They are listed below:
\begin{equation}
\beta_1=\mathrm{id}, 
\end{equation}
\begin{equation}
\begin{split}
\beta_2&=\begin{bmatrix}
1_2 & \\
& 0_2
\end{bmatrix} \otimes e_{11} \otimes 1_4 +\begin{bmatrix}
0_2 &\\
& 1_2
\end{bmatrix}\otimes e_{11}\otimes \begin{bmatrix}
1 & \\
& - 1_3
\end{bmatrix}+\\
&+\begin{bmatrix}
1& \\
& 0_3
\end{bmatrix}\otimes e_{22}\otimes 1_4 +\begin{bmatrix}
0 &\\
& 1_3
\end{bmatrix}\otimes e_{22}\otimes \begin{bmatrix}
 1_2 & \\
& - 1_2
\end{bmatrix},
\end{split}
\end{equation}
\begin{equation}
\beta_3 =\begin{bmatrix}
1_2 & \\
& -1_2
\end{bmatrix} \otimes e_{11} \otimes 1_4 +1_4\otimes e_{22}\otimes \begin{bmatrix}
 1_2 & \\
& -1_2
\end{bmatrix},
\end{equation}
\begin{equation}
\begin{split}
\beta_4&=\begin{bmatrix}
1_2 & \\
& 0_2
\end{bmatrix} \otimes e_{11} \otimes 1_4 +\begin{bmatrix}
0_2 &\\
& 1_2
\end{bmatrix}\otimes e_{11}\otimes \begin{bmatrix}
-1 & \\
& 1_3
\end{bmatrix}+ \\
&+\begin{bmatrix}
1& \\
& 0_3
\end{bmatrix}\otimes e_{22}\otimes \begin{bmatrix}
-1_2 & \\
&  1_2
\end{bmatrix} +\begin{bmatrix}
0 &\\
& 1_3
\end{bmatrix}\otimes e_{22}\otimes 1_4,
\end{split}
\end{equation}
\begin{equation}
\begin{split}
\beta_5&= \begin{bmatrix}
1_2 & \\
& 0_2
\end{bmatrix} \otimes e_{11} \otimes \begin{bmatrix}
 1& \\
& - 1_3
\end{bmatrix} +\begin{bmatrix}
0_2 &\\
& 1_2
\end{bmatrix}\otimes e_{11}\otimes 1_4+ \\
&+\begin{bmatrix}
1& \\
& 0_3
\end{bmatrix}\otimes e_{22}\otimes 1_4 +\begin{bmatrix}
0 &\\
& 1_3
\end{bmatrix}\otimes e_{22}\otimes \begin{bmatrix}
- 1_2 & \\
&  1_2
\end{bmatrix},
\end{split}
\end{equation}
\begin{equation}
\beta_6= \begin{bmatrix}
1_2 & \\
& -1_2
\end{bmatrix} \otimes e_{11} \otimes \begin{bmatrix}
1 & \\
& - 1_3
\end{bmatrix} +\begin{bmatrix}
1 & \\
& - 1_3
\end{bmatrix}\otimes e_{22}\otimes \begin{bmatrix}
1_2 & \\
& - 1_2
\end{bmatrix},
\end{equation}
\begin{equation}
\beta_7 =1_4 \otimes e_{11} \otimes \begin{bmatrix}
1 & \\
& -1_3
\end{bmatrix} +\begin{bmatrix}
1 &\\
& -1_3
\end{bmatrix}\otimes e_{22}\otimes 1_4,
\end{equation}
\begin{equation}
\begin{split}
\beta_8&=\begin{bmatrix}
1_2 & \\
& 0_2
\end{bmatrix} \otimes e_{11} \otimes \begin{bmatrix}
1 & \\
& - 1_3
\end{bmatrix} +\begin{bmatrix}
0_2 &\\
& -1_2
\end{bmatrix}\otimes e_{11}\otimes 1_4+ \\
&+\begin{bmatrix}
1& \\
& 0_3
\end{bmatrix}\otimes e_{22}\otimes \begin{bmatrix}
 1_2 & \\
& - 1_2
\end{bmatrix} +\begin{bmatrix}
0 &\\
& -1_3
\end{bmatrix}\otimes e_{22}\otimes 1_4.
\end{split}
\end{equation}
All of the above are $0$-cycles (more precisely: $\beta=\pi(1_2,0_2,\eta_1,\eta_2 1_3)J\pi(1_2,0_2,\eta_1,\eta_2 1_3)J^{-1}+\pi(0_2,1_2,\eta_3,\eta_4 1_3)J\pi(0_2,1_2,\eta_3,\eta_4 1_3)J^{-1}$ ), but only four of them are images of all non-zero elements of the algebra: $\beta_1$, $\beta_3$, $\beta_6$ and $\beta_7$. These are exactly the same cases we had in the case of the single term $0$-cycles. Moreover, an analogous computation to before shows that $\beta$s which are not of the one term type, do not allow for physically acceptable Dirac operators, basically because of its restrictions on the  $\cdots\otimes e_{11}\otimes \cdots$ and $\cdots \otimes e_{22}\otimes \cdots$ terms, i.e. the conclusion for section \ref{reducedPS} remains valid when more general $\beta$s are considered.

\subsubsection{The Standard Model}

Let us now discuss the generic case for the Standard Model. In \cite{bs} the one term case was discussed. Now, mirroring the above computation we can get the following family of possible $\beta$s:
\begin{equation}
\begin{split}
\beta &= \begin{bmatrix}
1 & \\
& 0_3
\end{bmatrix}\otimes e_{11}\otimes \begin{bmatrix}
\eta_1 & \\
& \eta_2 1_3
\end{bmatrix}+\begin{bmatrix}
0&&\\
& 1&\\
&&0_2
\end{bmatrix}\otimes e_{11}\otimes \begin{bmatrix}
\eta_3 & \\
&\eta_4 1_3
\end{bmatrix}+ \\ &+\begin{bmatrix}
0_2&\\
& 1_2
\end{bmatrix}\otimes e_{11}\otimes\begin{bmatrix}
\eta_5 &\\
&\eta_6 1_3 
\end{bmatrix}+
\begin{bmatrix}
1& \\
& 0_3
\end{bmatrix}\otimes e_{22}\otimes \begin{bmatrix}
\eta_1&&\\
&\eta_3 &\\
&& \eta_5 1_2 
\end{bmatrix}+\\&+ \begin{bmatrix}
0 &\\
& 1_3
\end{bmatrix}\otimes e_{22}\otimes \begin{bmatrix}
\eta_2 &&\\
& \eta_4 \\
&&\eta_6 1_2
\end{bmatrix},
\end{split}
\end{equation} 
where $\eta_i=\pm 1$, for $i=1,...,6$. This straightforward check shows that the only case in which $\beta$ is a $0$-cycle and allows for an extension of the Standard Model (in the previously discussed sense) is of the one term type and is exactly the same $\beta$ that prevented the existence of leptoquarks in \cite{bs}. Therefore, the validity of the conclusion therein is maintained when generalizing to allow, from the outset, for more general $\beta$s.

\section{The reduced Pati-Salam model and the Left-Right Symmetric models}
Notice that the class of reduced Pati-Salam models, which is the only allowed family of such models compatible with the existence of the pseudo-Riemannian structure $\beta$, may be considered as a generalized version of the so-called Left-Right Symmetric models (LRS). This class of models was broadly considered both from theoretical and phenomenological points of view - see e.g. \cite{pr}, \cite{mp},\cite{mp1}, \cite{bms}, \cite{phg}, \cite{sm} and \cite{sir}. In such models the gauge group is 
\begin{equation}
\mathrm{SU(2)}_R\times \mathrm{SU(2)}_L\times \mathrm{SU(3)}\times \mathrm{U(1)}_{B-L}.
\end{equation}
The chiral fermions consist of three families of quarks and leptons, and are given by
\begin{equation}
q_L=\left(1,2,3,\frac{1}{3}\right), \qquad q_R=\left(2,1,3,\frac{1}{3}\right),\qquad l_L=\left(1,2,1,-1\right), \qquad l_R=\left(2,1,1,-1\right),
\end{equation}   
where the parameters denote the quantum numbers under $\mathrm{SU(2)}_R$, $\mathrm{SU(2)}_L$, $\mathrm{SU(3)}$ and $\mathrm{U(1)}_{B-L}$ gauge groups, respectively \cite{pr}. 

The charge of a particle in such a model is defined as $Q=I_{3L}+I_{3R}+\frac{B-L}{2}$, where $I_3$ is the third component of an $\mathrm{SU(2)}$-isospin. 

The Left-Right Symmetric models were also considered from the point of view of noncommutative geometry, initially as possibly interesting examples for the Connes-Lott scheme, but later on also as possible extensions of the Standard Model --- see e.g. \cite{gir},\cite{hj},\cite{is} and \cite{ok}. The main interest was in the determination of whether, in this framework, such models provide a mechanism to explain the origin of parity symmetry breaking. In \cite{gir} it was argued that in the almost-commutative Yang-Mills-Higgs models, parity cannot be spontaneously broken. This followed from the requirement that Poincar\'{e} duality must be satisfied.  

The family of reduced Pati-Salam models generalizes both the Left-Right Symmetric Models and also the Chiral Electromagnetism Model \cite{gir}. The latter contains the $\mathrm{U(1)}_R\times \mathrm{U(1)}_L$ gauge group instead of the $\mathrm{SU(2)}$ ones. This theory played the role of a toy model for the application of the Connes-Lott scheme to LRS theories.

\section{Conclusions and outlook}

We have discussed a role which may be played by the existence of pseudo-Riemannian structures for the finite spectral triples associated with the family of Pati-Salam models. We have shown that their existence reduces the algebra to $\mathbb{H}_R\oplus\mathbb{H}_L\oplus \mathbb{C}\oplus M_3(\mathbb{C})$. Despite the fact that the existence of the additional grading as the shadow of a pseudo-Riemannian structure does not determine the Dirac operator uniquely, we have a huge reduction of the possible choices. 

We would like to stress that due to such a reduction, the family of Left-Right Symmetric models appears to be the interesting one. Due to the trend of relaxing some of the axioms of spectral geometry, it may be interesting to revisit the previous results for these classes of models in such a framework. We postpone this for future research.

\acknowledgments

The authors would like to thank A.Sitarz for his careful reading of the manuscript and helpful comments.


\end{document}